\documentclass[11pt]{article}

\usepackage{geometry}
\usepackage{enumitem}
\usepackage{cite}
\usepackage{graphicx}
\usepackage{subfigure}
\usepackage{amssymb}
\usepackage{amsmath}
\allowdisplaybreaks[1]

\usepackage{amsthm}

\newtheorem{theorem}{Theorem}
\newtheorem{lemma}[theorem]{Lemma}

\newtheorem{definition}[theorem]{Definition}

\makeatletter
\def\@endtheorem{\endtrivlist}
\makeatother

\newcounter{brule}
\addtocounter{brule}{-1}
\newenvironment{brule}{\refstepcounter{brule}\par\smallskip\noindent
\textbf{(B\arabic{brule})}\quad}{}
\newcommand{\currentrule}{B\arabic{brule}}

\newcommand{\Fmina}[1]{\mathcal{F}(#1)}
\newcommand{\Fmin}[1]{\mathcal{F}^*(#1)}
\newcommand{\Fmine}[1]{\mathcal{F}_\mathrm{edit}(#1)}
\newcommand{\Pmid}[1]{P_{\mathrm{mid}}(#1)}
\newcommand{\Pother}[1]{P_{\mathrm{other}}(#1)}

\begin{document}

\title{Faster algorithms for cograph edge modification problems}
\author{Dekel Tsur%
\thanks{Ben-Gurion University of the Negev.
Email: \texttt{dekelts@cs.bgu.ac.il}}}
\date{}
\maketitle

\begin{abstract}
In the \textsc{Cograph Deletion} (resp., \textsc{Cograph Editing}) problem
the input is a graph $G$ and an integer $k$,
and the goal is to decide whether there is a set of edges of size at most $k$
whose removal from $G$ (resp., removal and addition to $G$) results in a graph
that does not contain an induced path with four vertices.
In this paper we give algorithms for \textsc{Cograph Deletion} and
\textsc{Cograph Editing} whose running times are $O^*(2.303^k)$ and
$O^*(4.329^k)$, respectively.
\end{abstract}


\section{Introduction}
A graph $G$ is called a \emph{cograph} if it does not contain an induced
$P_4$, where $P_4$ is a path with 4 vertices.
In the \textsc{Cograph Deletion} (resp., \textsc{Cograph Editing}) problem
the input is a graph $G$ and an integer $k$,
and the goal is to decide whether there is a set of edges of size at most $k$
whose removal from $G$ (resp., removal and addition to $G$) results in
a cograph.
The \textsc{Cogaph Deletion} and \textsc{Cogaph Editing} problems can be solved
in $O^*(3^k)$ and $O^*(6^k)$ time (where $O^*(f(k)) = O(f(k)\cdot n^{O(1)})$), respectively, using simple branching algorithms~\cite{cai1996fixed}.
Nastos and Gao~\cite{nastos2012bounded} gave an algorithm with
$O^*(2.562^k)$ running time for \textsc{Cogaph Deletion} and
Liu et al.~\cite{liu2012complexity} gave an algorithm for
\textsc{Cogaph Editing} with $O^*(4.612^k)$ running time.
In this paper, we give algorithms for \textsc{Cograph Deletion} and
\textsc{Cogaph Editing}  whose time complexities are $O^*(2.303^k)$ and
$O^*(4.329^k)$, respectively.

\paragraph{Preleminaries}
For a set $S$ of vertices in a graph $G$, $G[S]$ is the subgraph
of $G$ induced by $S$ (namely, $G[S]=(S,E_S)$ where
$E_S = \{(u,v) \in E(G) \colon u,v \in S\}$).
For a graph $G = (V,E)$ and a set $F \subseteq E$ of edges,
$G-F$ is the graph $(V,E \setminus F)$.
For a graph $G = (V,E)$ and a set $F$ of pairs of vertices,
$G \triangle F$ is the graph $(V, (E \setminus F) \cup (F \setminus E))$.

For a graph $G$, a set $F$ of edges is called a \emph{deletion set} of $G$
if $G-F$ is a cograph.
A set $F$ of pairs of vertices is called an \emph{editing set} of $G$
if $G \triangle F$ is a cograph.

\section{Algorithm for Cogaph Deletion}\label{sec:deletion}
A graph $G$ is called \emph{$P_4$-sparse} if for every set $X$ of 5 vertices,
the graph $G[X]$ has at most one induced $P_4$.
A graph $G$ is $P_4$-sparse if and only if it does not contain one of the
graphs of Figure~\ref{fig:forbidden} as an induced subgraph.
Jamison and Olariu~\cite{jamison1992tree} showed that a $P_4$-sparse
graph admits a recursive decomposition.
To describe this decomposition, we first need the following definition.

\begin{figure}
\centering
\subfigure[$P_5$]{\includegraphics[scale=0.9]{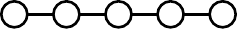}}
\subfigure[Pan\label{fig:pan}]{\includegraphics[scale=0.9]{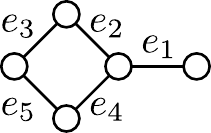}}
\subfigure[Kite\label{fig:kite}]{\includegraphics[scale=0.9]{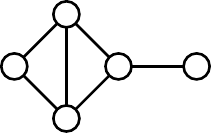}}
\subfigure[$C_5$\label{fig:C5}]{\includegraphics[scale=0.9]{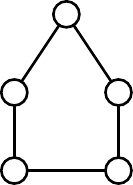}}
\subfigure[$\overline{P_5}$]{\includegraphics[scale=0.9]{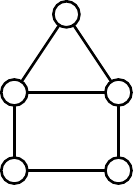}}
\subfigure[Fork\label{fig:fork}]{\includegraphics[scale=0.9]{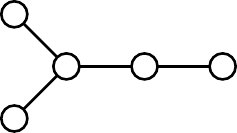}}
\subfigure[Co-pan]{\includegraphics[scale=0.9]{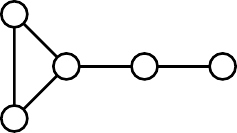}}
\caption{Forbidden induced subgraphs for $P_4$-sparse graphs.\label{fig:forbidden}}
\end{figure}

\begin{definition}\label{def:almost-spider}
A graph $G$ is a \emph{spider} if the vertices
of $G$ can be partitioned into disjoint sets $S$, $K$, and $R$ such that
\begin{enumerate}
\item
$|S| = |K| \geq 2$.
\item
$S$ is an independent set and $K$ is a clique.
\item
Every vertex in $R$ is adjacent to all the vertices in $K$ and
not adjacent to all the vertices in $S$.
\item\label{enu:legs}
There is a bijection $\varphi \colon S \to K$ such that
one of the following two cases occurs.
\begin{enumerate}
\item
$N(s) \cap K = \{ \varphi(s) \}$ for every $s \in S$.
\item
$N(s) \cap K = K \setminus \{ \varphi(s) \}$ for every $s \in S$.
\end{enumerate}
\end{enumerate}
\end{definition}

The recursive decomposition of $P_4$-sparse graphs is based on the following
theorem from~\cite{jamison1992tree}.
\begin{theorem}\label{thm:decomposition-1}
Let $G$ be a $P_4$-sparse graph with at least 2 vertices.
Then exactly one of the following cases occurs.
\begin{enumerate}
\item
$G$ is not connected.
\item
$\overline{G}$ is not connected.
\item
$G$ is a spider.
\end{enumerate}
\end{theorem}

For a graph $H$, let $\Fmina{H}$ be a set containing every inclusion minimal
deletion set of $H$.
The branching vector than corresponds to $\Fmina{H}$ is a vector that contains
the sizes of the sets in $\Fmina{H}$.
We now describe the algorithm of Nastos and Gao~\cite{nastos2012bounded}
for the \textsc{Cograph Deletion} problem.
The algorithm is a branching algorithm (cf.~\cite{cygan2015parameterized}).
The algorithm first repeatedly applies the following branching rule, until the
rule cannot be applied.
\begin{brule}
If $G$ is not $P_4$-sparse, find a set $X$ such that $G[X]$ is isomorphic
to one of the graphs in Figure~\ref{fig:forbidden}.
For every $F \in \Fmina{G[X]}$,
recurse on the instance $(G-F,k-|F|)$.\label{rule:forbidden}
\end{brule}


Let $\alpha(G)$ denote the minimum size of a deletion set of $G$.
When $G$ is $P_4$-sparse, the algorithm computes $\alpha(G)$ in polynomial
time and then returns whether $\alpha(G) \leq k$.
The computation of $\alpha(G)$ relies on the recursive decomposition of $G$.
\begin{lemma}\label{lem:alpha}
Let $G$ be a $P_4$-sparse graph.
If $G$ has less than 4 vertices, $\alpha(G) = 0$.
If $G$ (resp., $\overline{G}$) is not connected, let $C_1,\ldots,C_p$ be
the connected components of $G$ (resp., $\overline{G}$).
Then, $\alpha(G) = \sum_{i=1}^p \alpha(G[C_i])$.
Otherwise, $G$ is a spider and let $S,K,R$ be the corresponding partition of
the vertices of $G$.
Then, $\alpha(G) = \alpha(G[S \cup K])+\alpha(G[R])$,
where $\alpha(G[S\cup K])$ is either $|K|-1$ or $\binom{|K|}{2}$.
\end{lemma}

The worst case of Rule~(\currentrule) is when $X$ induces a pan.
In this case $\Fmina{G[X]} = \{\{e_1\}, \{e_2,e_4\}, \{e_2,e_5\}, \{e_3,e_4\},
\{e_3,e_5\} \}$,
where $e_1,\ldots,e_5$ are the edges of $G[X]$ according to
Figure~\ref{fig:pan}.
Thus, the branching vector in this case is $(1,2,2,2,2)$ and the branching
number is at most 2.562.
Therefore, the running time of the algorithm is $O^*(2.562^k)$.

The main idea behind our improved algorithm is as follows.
The proof of Theorem~\ref{thm:decomposition-1} is based on showing that
certain graphs with 6 to 8 vertices cannot occur in $G$ since
these graphs have induced subgraphs that are isomorphic to graphs in
Figure~\ref{fig:forbidden}.
Instead of destroying subgraphs of $G$ that are isomorphic to the graphs
of Figure~\ref{fig:forbidden},
our algorithm destroys the graphs considered in the proof of
Theorem~\ref{thm:decomposition-1}.
We then show that a graph that does not contain these graphs as induced
subgraphs admits a recursive
decomposition similar to the decomposition of $P_4$-sparse graphs,
and this decomposition can be used to solve the problem in polynomial time.

Before describing our algorithm, we need some definitions.
Suppose that $A$ is a set of vertices that induces a $P_4$.
Denote by $I(A)$ (resp., $T(A)$), the set of vertices
$v\in V(G) \setminus A$ such that $|N(v)\cap A| = 0$
(resp., $|N(v)\cap A| = 4$).
Let $P(A) = V(G)\setminus (A \cup I(A) \cup T(A))$.
In other words, $P(A)$ is the set of vertices $v \in V(G)\setminus A$
such that $1 \leq |N(v)\cap A| \leq 3$.
Let $\Pmid{A}$ be the set of vertices $v \in V(G)\setminus A$ such that
$|N(v)\cap A| = 2$ and $v$ is adjacent to the two internal vertices of the path
induced by $A$.
Let $\Pother{A} = P(A) \setminus \Pmid{A}$.
For a graph $H$, let $\Fmin{H} = \Fmina{H'}$, where $H'$ is an induced subgraph
of $H$ that contains at least one induced $P_4$
and such that the branching number corresponding to $\Fmina{H'}$ is minimized.

We now describe the branching rules of our algorithm.
In these rules, the algorithm finds a set $X$ of 6 to 8 vertices
that induces a certain graph and then it performs branching according to
$\Fmin{G[X]}$.
We note that the algorithm branches according to $\Fmin{G[X]}$ rather than
according to $\Fmina{G[X]}$ since in some cases the former branching is more
efficient
(for example, if $X$ induces a $P_7$ then the branching number corresponding to
$\Fmin{G[X]}$ is 2.303 while the branching number corresponding to
$\Fmina{G[X]}$ is 2.45).

The first three branching rules add restrictions on vertices in $\Pother{A}$.

\begin{brule}
Let $A$ be a set that induces a $P_4$ such that there are
distinct vertices $p \in \Pother{A}$ and $p' \in P(A)$ for which
$G[A\cup \{p,p'\}]$ is not isomorphic to either of the graphs in
in Figure~\ref{fig:Ps-P}.
For every $F \in \Fmin{G[A\cup \{p,p'\}]}$,
recurse on the instance $(G-F,k-|F|)$.\label{rule:Ps-P}
\end{brule}

To bound the branching number of Rule~(\currentrule), we used a Python script.
The script enumerates all possible cases for the graph $G[A\cup \{p,p'\}]$.
Denote by $a,b,c,d$ the path induced by $A$.
For the vertex $p$, the script enumerates all possible cases for $N(p) \cap A$.
After removing symmetric cases, there are 7 possible cases for $N(p) \cap A$:
$\{a\}$, $\{b\}$, $\{a,b\}$, $\{a,c\}$, $\{a,d\}$, $\{a,b,c\}$, and $\{a,c,d\}$.
For each case of $N(p) \cap A$,
the script enumerates all possible cases for $N(p') \cap A$
(here symmetric cases are not removed) and it also considers the cases
$(p,p') \notin E(G)$ and $(p,p') \in E(G)$.
For each possible case the scripts computes $\Fmin{G[A\cup \{p,p'\}]}$.
by enumerating the induced subgraphs of $G[A\cup \{p,p'\}]$.
For each induced subgraph $H$,
the script enumerates all subsets of $E(H)$ and checks
which of these sets are deletion sets.
The worst branching vector for Rule~(\currentrule) is $(1,2,2,2)$
and the branching number is at most $2.303$.
The script was also used to compute the branching numbers of the
other rules of the algorithm.

\begin{figure}
\centering
\subfigure[]{\includegraphics[scale=0.9]{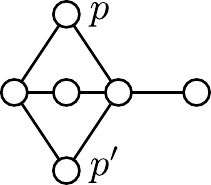}}
\qquad
\subfigure[\label{fig:E}]{\includegraphics[scale=0.9]{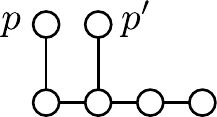}}
\caption{Unallowed induced subgraphs for
Rule~(B\ref{rule:Ps-P}).\label{fig:Ps-P}}
\end{figure}

\begin{brule}
Let $A$ be a set that induces a $P_4$ such that
there are non-adjacent vertices $p \in \Pother{A}$ and $t \in T(A)$.
For every $F \in \Fmin{G[A\cup \{p,t\}]}$,
recurse on the instance $(G-F,k-|F|)$.\label{rule:Ps-T}
\end{brule}


\begin{brule}
Let $A$ be a set that induces a $P_4$ such that
there are adjacent vertices $p \in \Pother{A}$ and $i \in I(A)$
for which $G[A\cup \{p,i\}]$ is not isomorphic to the graph in
Figure~\ref{fig:Ps-I}.
For every $F \in \Fmin{G[A\cup \{p,i\}]}$,
recurse on the instance $(G-F,k-|F|)$.\label{rule:Ps-I}
\end{brule}

\begin{figure}
\centering
\includegraphics[scale=0.9]{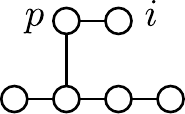}
\caption{Unallowed induced subgraph for
Rule~(B\ref{rule:Ps-I}).\label{fig:Ps-I}}
\end{figure}


The last rule handles cases in which the proof of
Theorem~\ref{thm:decomposition-1} relies on the fact that $G$ is $P_4$-sparse.
For example, Observation~2.3 in~\cite{jamison1992tree}
states that there is no vertex $t \in T(A)$ and vertices
$u,v \in I(A) \cup P(A)$ such that $E(G[\{u,v,t\}]) = \{(u,v),(u,t)\}$.
This observation is proved by showing that if such vertices exist, $G$ has
an induced subgraph isomorphic to a fork (see Figure~\ref{fig:fork}).
Therefore, we add a rule that is applicable if there are such
vertices $t,u,v$. The rule performs branching according to
$\Fmin{G[A\cup \{t,u,v\}]}$.
Since the proof of Theorem~\ref{thm:decomposition-1} can consider the complement
graph of $G$, we also need to add a complement rule that is applicable
if the former rule is applicable in $\overline{G}$.
In other words, the new rule is applicable if there are vertices $i \in I(A)$
and $u,v \in T(A) \cup P(A)$ such that $E(G[\{u,v,i\}]) = \{(v,i)\}$.

\begin{brule}
Let $A$ be a set that induces a $P_4$ such that one of the following
cases occurs.
\begin{enumerate}[itemsep=0pt]
\item\label{enu:2.3}
There are vertices $u,v \in I(A) \cup P(A)$ and $x \in T(A)$ such that
$E(G[\{u,v,x\}]) = \{(u,v),(u,x)\}$ (Observation~2.3 in~\cite{jamison1992tree}).
\item\label{enu:2.3b}
There are vertices $u,v \in T(A) \cup P(A)$ and  $x \in I(A)$ such that
$E(G[\{u,v,x\}]) = \{(v,x)\}$ (Observation~2.3).\label{rule:theorem-case-2}
\item\label{enu:2.4}
There are vertices $v \in I(A) \cup P(A)$ and $x,y \in T(A)$ such that
$E(G[\{v,x,y\}]) = \{(v,x)\}$ (Observation~2.4).
\item\label{enu:2.4b}
There are vertices $v \in T(A)\cup P(A)$ and $x,y \in I(A)$ such that
$E(G[\{v,x,y\}]) = \{(v,y),(x,y)\}$ (Observation~2.4).
\item\label{enu:2.6}
There are vertices $v \in P(A)$ and $x,y \in T(A)$ such that
$v$ is not adjacent to $x,y$ (Observation~2.6).
\item\label{enu:2.6b}
There are vertices $v \in P(A)$ and $x,y \in I(A)$ such that
$v$ is adjacent to $x,y$ (Observation~2.6).
\item\label{enu:2.7}
There are vertices $u,v \in P(A)$ and $x \in T(A)$ such that
$x$ is not adjacent to $u,v$ (Observation~2.7).
\item\label{enu:2.7b}
There are vertices $u,v \in P(A)$ and $x \in I(A)$ such that
$x$ is adjacent to $u,v$ (Observation~2.7).
\item\label{enu:2.14}
There are vertices $v \in P(A)$, $x \in T(A)$, and $y \in I(A)$ such that
$E(G[\{v,x,y\}]) = \{(x,y)\}$ (Observation~2.14).
\item\label{enu:2.14b}
There are vertices $v \in P(A)$, $x \in I(A)$, and $y \in T(A)$ such that
$E(G[\{v,x,y\}]) = \{(v,x),(v,y)\}$ (Observation~2.14).
\item\label{enu:2.15}
There are vertices $v \in P(A)$, $x \in T(A)$, and $y \in I(A)$ such that
$E(G[\{v,x,y\}]) = \{(v,y)\}$ (Fact~2.15).
\item\label{enu:2.15b}
There are vertices $v \in P(A)$, $x \in I(A)$, and $y \in T(A)$ such that
$E(G[\{v,x,y\}]) = \{(v,x),(x,y)\}$ (Fact~2.15).
\item\label{enu:2.15c}
There are vertices $v \in P(A)$, $x,y \in T(A)$, and $z \in I(A)$, such that
$E(G[\{v,x,y,z\}]) = \{(v,x),(x,y),(y,z)\}$ (Fact~2.15).
\item\label{enu:2.15d}
There are vertices $v \in P(A)$, $x,y \in I(A)$, and $z \in T(A)$ such that
$E(G[\{v,x,y,z\}]) = \{(v,y),(v,z),(x,z)\}$ (Fact~2.15).
\end{enumerate}
For every $F \in \Fmin{G[A \cup B]}$,
where $B$ is a set containing the vertices mentioned in the cases above,
(namely, $B = \{u,v,x\}$ in Cases~\ref{enu:2.3}, \ref{enu:2.3b}, \ref{enu:2.7},
and~\ref{enu:2.7b},
$B = \{v,x,y\}$ in Cases~\ref{enu:2.4}, \ref{enu:2.4b}, \ref{enu:2.6},
\ref{enu:2.6b}, \ref{enu:2.14}, \ref{enu:2.14b}, \ref{enu:2.15},
and~\ref{enu:2.15b},
and $B = \{v,x,y,z\}$ in Cases~\ref{enu:2.15c} and~\ref{enu:2.15d}).
recurse on the instance $(G-F,k-|F|)$.\label{rule:theorem}
\end{brule}

The branching numbers of Rule~(B\ref{rule:Ps-T}), Rule~(B\ref{rule:Ps-I})
and Rule~(\currentrule) are at most 2.27, 2.303, and 2.21, respectively.


We now show that a graph in which the branching rules cannot be applied
has a recursive decomposition.
\begin{theorem}\label{thm:decomposition-2}
Let $G$ be a graph with at least 7 vertices in which
Rules~(B1)--(B\ref{rule:theorem}) cannot be applied.
Then one of the following cases occurs.
\begin{enumerate}
\item
$G$ is not connected.
\item
$\overline{G}$ is not connected.
\item
$G$ is a spider.
\item\label{case:bipartite}
$G$ is a bipartite graph with parts $X$ and $Y$ such that (a) $|X|=2$
(b) There is a vertex $y \in Y$ such that $y$ is adjacent to exactly one vertex
in $X$ and every vertex in $Y\setminus \{y\}$ is adjacent to the two vertices
of $X$.
\end{enumerate}
\end{theorem}
\begin{proof}
Suppose that $G$ is a graph such that $G$ and $\overline{G}$ are connected.
We choose a set $A$ that induces a $P_4$ such that $|P(A)|$ is maximized.
Denote the vertices of the path by $a,b,c,d$.

Recall that Theorem~\ref{thm:decomposition-1} requires the graph $G$ to be
$P_4$-sparse. The proof of Theorem~\ref{thm:decomposition-1} relies on this
requirement in the following places:
(1) in the proofs of Observations 2.3, 2.4, 2.6, 2.7, 2.14, and~2.15.
(2) the requirement that $G$ is $P_4$ sparse implies that
$\Pother{A} = \emptyset$ and the emptiness of $\Pother{A}$ is used in the proof.
Since Rule~(B\ref{rule:theorem}) cannot be applied, the observations mentioned
above remain true.
If $\Pother{A} = \emptyset$ then the arguments used in the proof of
Theorem~\ref{thm:decomposition-1} are also true here, and therefore
$G$ is a spider.

Now suppose that $\Pother{A} \neq \emptyset$.
Since Rule~(B\ref{rule:Ps-P}) cannot be applied, $\Pmid{A} = \emptyset$.
The proof of Theorem~\ref{thm:decomposition-1} shows the following
properties
(More precisely, the proof of Theorem~\ref{thm:decomposition-1} shows that the
first property below is satisfied in either $G$ or $\overline{G}$.
The second property below is satisfied in $G$ if and only if the first property
is satisfied in $\overline{G}$).
\begin{enumerate}
\item\label{prop:T-P}
If $T(A) \neq \emptyset$ then there is an injective mapping
$\varphi \colon T(A) \to P(A)$ such that
$N(t) \cap P(A) = P(A) \setminus \{\varphi(t)\}$ for every $t \in T(A)$.
\item\label{prop:I-P}
If $I(A) \neq \emptyset$ then there is an injective mapping
$\varphi \colon I(A) \to P(A)$ such that
$N(i) \cap P(A) = \{\varphi(i)\}$ for every $i \in I(A)$.
\end{enumerate}
These properties are also true here, since the proof of these properties
does not rely on the emptiness of $\Pother{A}$.

We now claim that $T(A) = \emptyset$.
Suppose conversely that there is a vertex $t \in T(A)$.
By Property~\ref{prop:T-P}, there is a vertex $p \in P(A) = \Pother{A}$
such that $t$ is not adjacent to $p$.
This is a contradiction to the assumption that Rule~(B\ref{rule:Ps-T})
cannot be applied.
Therefore, $T(A) = \emptyset$.

Suppose that $I(A) = \emptyset$.
If $|P(A)| \leq 2$ then $G$ has at most 6 vertices and we are done.
Otherwise, since Rule~(B\ref{rule:Ps-P}) cannot be applied,
without loss of generality $N(p) = \{a,c\}$ for every $p \in P(A)$
and $P(A)$ is an independent set.
Therefore, $G$ satisfies Case~\ref{case:bipartite} of the theorem
with $X = \{a,c\}$ and $Y = V(G) \setminus X$.

Now suppose that $I(A) \neq \emptyset$.
Let $i \in I(A)$.
By Property~\ref{prop:I-P}, there is a vertex $p \in P(A)$ such that
$N(i) \cap P(A) = \{p\}$.
Since Rule~(B\ref{rule:Ps-I}) cannot be applied,
without loss of generality, $N(p) \cap A = \{b\}$.
We claim that $|I(A)| = 1$.
Suppose conversely that $|I(A)|>1$ and let $i' \in I(A) \setminus \{i\}$.
By Property~\ref{prop:I-P}, $i'$ is adjacent to a vertex
$p' \in P(A) \setminus \{p\}$.
Since Rule~(B\ref{rule:Ps-I}) cannot be applied, either $N(p') \cap A = \{b\}$
or $N(p') \cap A = \{c\}$.
In both cases we obtain a contradiction to the assumption that
Rule~(B\ref{rule:Ps-P}) cannot be applied.
Therefore, $|I(A)| = 1$.
Since Rule~(B\ref{rule:Ps-P}) cannot be applied, either
$P(A) = \{p\}$, or $P(A) = \{p,p'\}$ and $N(p') \cap A = \{a\}$.
In the latter case we obtain that case~\ref{rule:theorem-case-2} of
Rule~(B\ref{rule:theorem}) can be applied on $i,p,p'$
(recall that $N(i) \cap P(A) = \{p\}$), a contradiction.
Therefore, $P(A) = \{p\}$. By Property~\ref{prop:I-P}, $I(A) = \{i\}$,
and thus $G$ contains 6 vertices.
\end{proof}

Given a graph $G$ in which Rules~(B1)--(B\ref{rule:theorem})
cannot be applied, $\alpha(G)$ can be computed in polynomial times as follows.
If $G$ has at most 6 vertices, compute $\alpha(G)$ by enumerating all subsets
of $E(G)$.
If $G$ satisfies Case~\ref{case:bipartite} of Theorem~\ref{thm:decomposition-2}
then $\alpha(G) = 1$.
Otherwise, we have that either $G$ is not connected, $\overline{G}$ is not
connected, or $G$ is a spider.
In each of these cases $\alpha(G)$ can be computed recursively according to
Lemma~\ref{lem:alpha}.

\begin{theorem}
There is an $O^*(2.303^k)$-time algorithm for \textsc{Cograph Deletion}.
\end{theorem}
\begin{proof}
Every application of a branching rule decreases the number of edges of the
graph by at least 1. Therefore, the depth of the recursion tree is at most
$|E|$.
By Theorem~\ref{thm:decomposition-2}, an instance in which the
branching rules cannot be applied can be solved in polynomial time.
Since each rule can be applied in polynomial time and since the branching rules
have branching numbers at most 2.303, it follows that the time complexity of the algorithm is $O^*(2.303^k)$.

The correctness of the algorithm follows from the safeness of the rules
(which follows from the definition of $\Fmin{H}$)
and from the fact that the algorithm terminates on every input.
\end{proof}

\section{Algorithm for Cograph Editing}\label{sec:editing}

The algorithm for \textsc{Cograph Editing} is similar to the
algorithm for \textsc{Cograph Deletion}. The differences are as follows.

\begin{enumerate}
\item
In each branching rule, the algorithm recurses on
$(G\triangle F, k-|F|)$ for every $F \in \Fmine{H}$ for the corresponding
graph $H$,
where $\Fmine{H}$ is a set containing every inclusion minimal editing
set of $H$.
\item
Rule~(B\ref{rule:Ps-P}) is applied also when $G[A\cup \{p,p'\}]$ is isomorphic
to a graph in Figure~\ref{fig:Ps-P}.
Similarly, Rule~(B\ref{rule:Ps-I}) is applied also when $G[A\cup \{p,i\}]$ is
isomorphic to the graph in Figure~\ref{fig:Ps-I}.
\item
If $G$ is a graph in which Rules (B\ref{rule:Ps-P})--(\currentrule) cannot be
applied, either $G$ contains at most 5 vertices,
$G$ is not connected, $\overline{G}$ is not connected, or $G$ is a spider.
Therefore, the minimum size of an editing set of $G$ can be computed in
polynomial time using the algorithm for $P_4$-sparse graphs of
Liu et al.~\cite{liu2012complexity}.
\end{enumerate}

Using an automated analysis as in Section~\ref{sec:deletion},
we obtain that the branching numbers of Rules
(B\ref{rule:Ps-P})--(B\ref{rule:Ps-I}) are at most 4.313,
and the branching number of Rule~(B\ref{rule:theorem}) is at most 4.329.
We obtain the following theorem.

\begin{theorem}
There is an $O^*(4.329^k)$-time algorithm for \textsc{Cograph Editing}.
\end{theorem}

\bibliographystyle{abbrv}
\bibliography{cograph}

\begin{thebibliography}{1}

\bibitem{cai1996fixed}
L.~Cai.
\newblock Fixed-parameter tractability of graph modification problems for
  hereditary properties.
\newblock {\em Information Processing Letters}, 58(4):171--176, 1996.

\bibitem{cygan2015parameterized}
M.~Cygan, F.~V. Fomin, {\L}.~Kowalik, D.~Lokshtanov, D.~Marx, M.~Pilipczuk,
  M.~Pilipczuk, and S.~Saurabh.
\newblock {\em Parameterized algorithms}.
\newblock Springer, 2015.

\bibitem{jamison1992tree}
B.~Jamison and S.~Olariu.
\newblock A tree representation for p4-sparse graphs.
\newblock {\em Discrete Applied Mathematics}, 35(2):115--129, 1992.

\bibitem{liu2012complexity}
Y.~Liu, J.~Wang, J.~Guo, and J.~Chen.
\newblock Complexity and parameterized algorithms for cograph editing.
\newblock {\em Theoretical Computer Science}, 461:45--54, 2012.

\bibitem{nastos2012bounded}
J.~Nastos and Y.~Gao.
\newblock Bounded search tree algorithms for parametrized cograph deletion:
  efficient branching rules by exploiting structures of special graph classes.
\newblock {\em Discrete Mathematics, Algorithms and Applications}, 4(01), 2012.

\end{thebibliography}

\end{document}